\documentclass[11pt]{article}

\usepackage[a4paper,margin=1in]{geometry}
\usepackage{amsmath,amssymb,amsthm,mathtools}
\usepackage{graphicx}
\usepackage{hyperref}
\usepackage{microtype}
\usepackage{enumitem}
\usepackage{tikz} %
\usetikzlibrary{calc,arrows.meta}

\newtheorem{theorem}{Theorem}[section]

\newtheorem{lemma}[theorem]{Lemma}

\theoremstyle{definition}

\newtheorem{remark}[theorem]{Remark}

\title{On Triangulations Generated by the Largest-Angle $n$-Section Algorithm}
\author{Jérôme Michaud and Sergey Korotov}
\date{\today}

\begin{document}

\maketitle
\begin{center}
Department of Business and Mathematics, IEM, Mälardalen University, Västerås, Sweden \newline
emails: jerome.michaud@mdu.se,  sergey.korotov@mdu.se 
\end{center}

\begin{abstract}
We define a mesh refinement algorithm based on the rule of dividing the largest angles of  triangular elements of planar partitions in focus into $n$ equal parts, and analyse the (geometric) properties of triangulations generated by this technique.  This largest-angle $n$-section rule is compared with the classical longest-edge $n$-section rule, where it is the longest edges which are split into $n$ equal parts.  The longest-edge bisection and trisection are known to produce nondegenerate triangulations (possibly with hanging nodes), but the longest-edge $n$-sections with $n\geq 4$ always produce (infinite) sequences of triangles with minimum angles tending to zero (moreover, their relevant maximum angles tend to $\pi$), thus breaking the minimum and maximum angle conditions.  We show that this degeneration effect is not a consequence of $n$-section itself.  For every $n\geq 2$, the largest-angle $n$-sections produce partitions satisfying the minimum angle condition (and, therefore, the  maximum angle condition).  More precisely, if the initial triangle has its smallest angle $\gamma_0>0$, then all descendant triangles have angles bounded below by
\(
        m_n=\min\left\{\gamma_0,\frac{\pi}{3n}\right\},
\)
and, correspondingly, bounded above by $\pi-2m_n<\pi$.  We also show that the recursive largest-angle $n$-section algorithm always produces a family of triangular partitions, i.e. the maximum diameter of level-$k$ descendants tends to zero as $k \to \infty$.
\end{abstract}

\medskip


\medskip


\section{Introduction}

Mesh refinement procedures are central in various adaptive finite element computations.  A basic quality requirement is that the refinement process should not create arbitrarily degenerating mesh elements.  In two dimensions, this is often expressed through the minimum angle condition: there should be a positive lower bound for all angles appearing in the refinement process.  A weaker, but still important, requirement is the maximum angle condition, which requires that all maximum angles produced remain bounded away from $\pi$.

The longest-edge bisection procedure is the classical example of a stable refinement rule.  Rosenberg and Stenger proved that longest-edge bisection preserves a positive lower bound on angles in terms of the minimum angle of the initial triangle \cite{RosenbergStenger1975}.  Longest-edge trisection is also known to be nondegenerate; Plaza, Suarez, Padron, and Falcon proved a positive minimum-angle bound for the longest-edge trisection \cite{PlazaSuarezPadronFalcon2010}.  The behavior changes for higher sections.  In fact, Suarez, Moreno, Abad, and Plaza proved that longest-edge $n$-sections produce degenerating triangles for $n\geq 4$, in the sense that one can generate triangles whose minimum angles tend to zero \cite{SuarezMorenoAbadPlaza2012}.  Later, Perdomo and Plaza presented a shorter proof of the same degeneracy phenomenon using a complex-variable normalization of triangle shapes \cite{PerdomoPlaza2012}.  Korotov, Plaza, and Suarez showed that the conforming longest-edge $n$-section algorithm violates even the weaker maximum angle condition for $n\geq 4$ \cite{KorotovPlazaSuarez2015}.

These results suggest a natural question: is the degeneracy for large $n$ caused by the act of dividing into many parts, or by the fact that the divided object is the longest edge?  The purpose of this paper is to show that the latter is the decisive feature.  If, instead of dividing the longest edge, one divides the largest angle, then a simple regularity mechanism holds for every $n\geq 2$.

The largest-angle bisection was recently studied by Ismailescu, Kim, Kim, and Lee, who proved, among other properties, that the procedure has the expected non-degeneracy behavior \cite{IsmailescuKimKimLee2019}.  The present paper records the observation that the same regularity mechanism extends directly to largest-angle $n$-section for every $n\geq 2$.  The key point is that the (largest) angle being divided is always at least $\pi/3$.  Hence each newly created angle is at least $\pi/(3n)$.

  In Section~\ref{sec:definition}, we define the largest-angle $n$-section directly in terms of angle triples.  Section~\ref{sec:main} contains the proofs for validity of the minimum/maximum angle conditions, and for the convergence of element diameters to zero.  Section~\ref{sec:conclusion} gives a brief conclusion.

\section{Largest-angle $n$-section rule}\label{sec:definition}

Let $T$ be a nondegenerate triangle with angles
\[
        \alpha,\beta,\gamma>0,
        \qquad
        \alpha+\beta+\gamma=\pi.
\]
A largest angle is any one of the angles equal to
\[
        \max\{\alpha,\beta,\gamma\}.
\]
If there is more than one largest angle, one of them is chosen arbitrarily.  This choice will not affect the estimates below.

Suppose, after relabelling the vertices of $T$ if necessary, that
\[
        \alpha=\max\{\alpha,\beta,\gamma\}.
\]
The largest-angle $n$-section of $T$ is obtained by dividing the angle $\alpha$ into $n$ equal parts and joining the corresponding division rays to the opposite side.  This produces $n$ child triangles.  Reading off the angles of these children gives the following angle triples:
\begin{equation}\label{eq:child-angle-triples}
        T_j=
        \left(
        \frac{\alpha}{n},
        \beta+\frac{j\alpha}{n},
        \gamma+\frac{(n-1-j)\alpha}{n}
        \right),
        \qquad
        j=0,1,\ldots,n-1.
\end{equation}
Indeed, the $j$th child contains one sector angle $\alpha/n$, receives $j$ adjacent sectors on the side of the angle $\beta$, and receives the remaining $n-1-j$ adjacent sectors on the side of the angle $\gamma$.  The sum of the three entries in \eqref{eq:child-angle-triples} is
\[
        \frac{\alpha}{n}
        +\beta+\frac{j\alpha}{n}
        +\gamma+\frac{(n-1-j)\alpha}{n}
        =
        \alpha+\beta+\gamma
        =
        \pi,
\]
so each triple is really the angle triple of a triangle, see figure \ref{fig:LA-n-section}.

\begin{figure}[t] 
\centering 
\includegraphics[width=\textwidth]{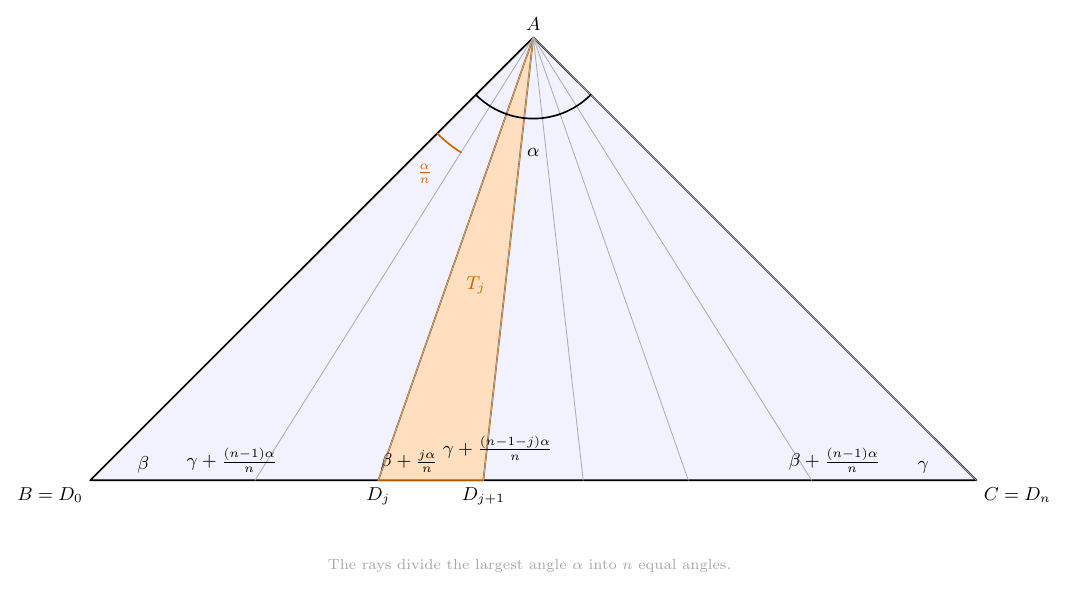}
 
\caption{The largest-angle $n$-section of the triangle $ABC$, whose largest angle $\alpha$ is divided into $n$ equal parts. The child triangle $T_j$, between the rays $AD_j$ and $AD_{j+1}$, has an angle triple \( \left( \frac{\alpha}{n}, \beta+\frac{j\alpha}{n}, \gamma+\frac{(n-1-j)\alpha}{n} \right). \)} \label{fig:LA-n-section} 
\end{figure}

Starting from a single triangle $T_0$, recursive largest-angle $n$-section means applying the above rule to every triangle in the current generation.  We denote by $\mathcal T_k$ the set of $n^k$ triangles obtained after $k$ generations.  More generally, if one starts from a finite collection of initial triangles (initial triangulation), the same estimates apply to all descendants of all initial triangles.  No conformity assertion for global triangular partitions is needed in this note; the estimates are purely element-wise.

\section{Regularity and diameter convergence properties}\label{sec:main}

We now prove the main result, which contrasts sharply with the known behavior of the longest-edge $n$-section for $n\geq 4$.

\begin{theorem}[Regularity of the largest-angle $n$-section]\label{thm:main}
Let $n\geq 2$, and let $T_0$ be a nondegenerate triangle with the smallest angle $\gamma_0>0$.  Let $\mathcal T_k$ be the set of descendants ($k$-th generation) obtained after $k$ steps of recursive largest-angle $n$-section.  Define
\begin{equation}\label{eq:mn}
        m_n
        =
        \min\left\{\gamma_0,\frac{\pi}{3n}\right\}.
\end{equation}
Then every angle of every triangle in every generation produced satisfies
\begin{equation}\label{eq:min-angle-condition}
        \theta\geq m_n.
\end{equation}
Consequently, any ever produced angle satisfies also the following upper bound
\begin{equation}\label{eq:max-angle-condition}
        \theta \leq \pi-2m_n<\pi.
\end{equation}
Thus, the largest-angle $n$-section algorithm always produces sets of triangles satisfying both, the minimum angle condition and the maximum angle condition.
\label{th:LAB-n-regularity}
\end{theorem}

\begin{proof}
It is enough to prove that the lower bound is preserved from one generation to the next.

Let $T$ be a parent triangle whose angles are $\alpha,\beta,\gamma$, and suppose that
\[
        \alpha=\max\{\alpha,\beta,\gamma\}.
\]
Since $\alpha$ is the largest angle of a triangle,
\begin{equation}\label{eq:largest-at-least-pi3}
        \alpha\geq \frac{\pi}{3}.
\end{equation}
The angles of the $j$-th child $T_j$ are given by \eqref{eq:child-angle-triples}.  Hence every angle of this child is bounded below by
\[
        \min\left\{
        \frac{\alpha}{n},
        \beta,
        \gamma
        \right\}.
\]
Using \eqref{eq:largest-at-least-pi3}, we get
\[
        \frac{\alpha}{n}
        \geq
        \frac{\pi}{3n}.
\]
Therefore every angle of every child is at least
\[
        \min\left\{\beta,\gamma,\frac{\pi}{3n}\right\}.
\]
In particular, if every angle of the parent is at least $m_n$, then every angle of every child is also at least $m_n$, because $m_n\leq \pi/(3n)$.

The initial triangle has all angles at least $\gamma_0$, and by definition $m_n\leq \gamma_0$.  Thus the bound \eqref{eq:min-angle-condition} follows by induction over the generations.

Finally, if all three angles of a triangle are at least $m_n$, then the largest angle is at most
\[
        \pi-m_n-m_n
        =
        \pi-2m_n.
\]
Since $m_n>0$, this gives \eqref{eq:max-angle-condition}.
\end{proof}

\begin{remark}
The lower bound \eqref{eq:mn} depends on the initial triangle and on $n$.  It is not uniform over all triangles in the initial triangulation as the initial angles $\gamma_0$ can vary considerably.  It is, however, uniform over the entire refinement tree generated from a fixed nondegenerate initial triangle.
\end{remark}

We shall now prove that recursive largest-angle \(n\)-section produces a (refining) family of triangular partitions, in the sense that the maximum diameter of the level-\(k\) descendants tends to zero. The proof uses an explicit area contraction estimate. \begin{lemma}[Area ratios for largest-angle \(n\)-section] Let \(T = ABC\) be a triangle with angles \[ \alpha,\beta,\gamma>0, \qquad \alpha+\beta+\gamma=\pi, \] and suppose that \[ \alpha=\max\{\alpha,\beta,\gamma\}. \] Let \(T_0,\ldots,T_{n-1}\) be the \(n\) children triangles obtained after dividing \(\alpha\) into \(n\) equal parts. Assume moreover that all angles of \(T\) are at least \(m>0\). Then every child $T_j \ (j = 0, \dots , n-1)$ satisfies \[ \frac{\operatorname{area}(T_j)}{\operatorname{area}(T)} \geq \delta_{m,n}, \qquad \delta_{m,n} := \frac{\pi}{3n}\sin^2 m . \] Consequently, \[ \operatorname{area}(T_j) \leq q_{m,n}\operatorname{area}(T), \qquad q_{m,n} := 1-(n-1)\frac{\pi}{3n}\sin^2 m <1 . \] 
\label{lm:area-ratios}
\end{lemma} 
\begin{proof} Let the largest angle \(\alpha\) be located at the vertex \(A\), and let the opposite side be \(BC\). Write the remaining angles as \(\beta\) at \(B\) and \(\gamma\) at \(C\). The \(n\)-section rays from \(A\) meet \(BC\) at points \[ D_0=B,\ D_1,\ldots,D_{n-1},\ D_n=C. \] The child \(T_j\) is the triangle \(AD_jD_{j+1}\). For \(0\leq \theta\leq \alpha\), let \(D(\theta)\) be the point where the ray from \(A\) making angle \(\theta\) with \(AB\) meets \(BC\). Define \[ F(\theta) = \frac{|BD(\theta)|}{|BC|}. \] Since all children have the same altitude from \(A\) to the line \(BC\), the area ratio of the \(j\)-th child is \[ \frac{\operatorname{area}(T_j)}{\operatorname{area}(T)} = F\left(\frac{(j+1)\alpha}{n}\right) - F\left(\frac{j\alpha}{n}\right). \] We compute \(F\). By the sine rule in the triangle \(ABD(\theta)\), and then in the parent triangle \(ABC\), \[ \frac{|BD(\theta)|}{\sin\theta} = \frac{|AB|}{\sin(\beta+\theta)}, \qquad |AB| = |BC|\frac{\sin\gamma}{\sin\alpha}. \] Hence \[ F(\theta) = \frac{\sin\gamma\,\sin\theta} {\sin\alpha\,\sin(\beta+\theta)}. \] Differentiating gives \[ F'(\theta) = \frac{\sin\beta\,\sin\gamma} {\sin\alpha\,\sin^2(\beta+\theta)}. \] Since all angles of \(T\) are at least \(m\), and since \(\alpha\) is the largest angle, the two smaller angles \(\beta\) and \(\gamma\) lie in \([m,\pi/2]\). Therefore \[ \sin\beta\geq \sin m, \qquad \sin\gamma\geq \sin m. \] Also \(\sin\alpha\leq 1\) and \(\sin^2(\beta+\theta)\leq 1\). Thus \[ F'(\theta)\geq \sin^2 m \qquad \text{for }0\leq \theta\leq \alpha. \] It follows that \[ \frac{\operatorname{area}(T_j)}{\operatorname{area}(T)} = \int_{j\alpha/n}^{(j+1)\alpha/n} F'(\theta)\,d\theta \geq \frac{\alpha}{n}\sin^2 m. \] Since \(\alpha\) is the largest angle, \(\alpha\geq \pi/3\), and therefore \[ \frac{\operatorname{area}(T_j)}{\operatorname{area}(T)} \geq \frac{\pi}{3n}\sin^2 m = \delta_{m,n}. \] The \(n\) child area ratios sum to \(1\). Since every child has area ratio at least \(\delta_{m,n}\), any one child has area ratio at most \[ 1-(n-1)\delta_{m,n}. \] Thus, \[ \operatorname{area}(T_j) \leq \left(1-(n-1)\frac{\pi}{3n}\sin^2 m\right) \operatorname{area}(T). \] This proves the claim. \end{proof} \begin{lemma}[Area controls diameter under the minimum angle condition] Let \(T\) be a triangle whose angles are all at least \(m>0\). Then \[ \operatorname{area}(T) \geq \frac12\,\operatorname{diam}(T)^2\sin^2 m. \] 
\label{lm:area-controls}
\end{lemma} \begin{proof} Let \(h=\operatorname{diam}(T)\) be the length of the longest side of \(T\), and let \(\Theta\) be the angle opposite this side. Let the two remaining angles be \(\Phi\) and \(\Psi\). By the sine rule, the two sides adjacent to \(\Theta\) have lengths \[ h\frac{\sin\Phi}{\sin\Theta} \qquad\text{and}\qquad h\frac{\sin\Psi}{\sin\Theta}. \] Hence \[ \operatorname{area}(T) = \frac12 h^2 \frac{\sin\Phi\sin\Psi}{\sin\Theta}. \] Since \(\Phi,\Psi\geq m\) and \(\sin\Theta\leq 1\), we obtain \[ \operatorname{area}(T) \geq \frac12 h^2\sin^2 m. \] \end{proof} \begin{theorem}[Diameter convergence] Let \(n\geq 2\), and let \(T_0\) be a nondegenerate initial triangle with smallest angle \(\gamma_0>0\). Let \(\mathcal T_k\) denote the family of triangles obtained after \(k\) generations of recursive largest-angle \(n\)-section. Set \[ m_n = \min\left\{\gamma_0,\frac{\pi}{3n}\right\} \] and \[ q_n = 1-(n-1)\frac{\pi}{3n}\sin^2 m_n. \] Then \(q_n<1\), and every \(T\in\mathcal T_k\) satisfies \[ \operatorname{diam}(T) \leq \frac{\sqrt{2\,\operatorname{area}(T_0)}}{\sin m_n} q_n^{k/2}. \] In particular, \[ \max_{T\in\mathcal T_k}\operatorname{diam}(T) \longrightarrow 0 \qquad \text{as } k\to\infty. \] \end{theorem} \begin{proof} By the minimum angle estimate from Theorem~\ref{th:LAB-n-regularity}, every descendant triangle has all angles at least \(m_n\). Therefore the explicit area contraction Lemma~\ref{lm:area-ratios} applies at every refinement step with \(m=m_n\). Hence every \(T\in\mathcal T_k\) satisfies \[ \operatorname{area}(T) \leq q_n^k \operatorname{area}(T_0). \] On the other hand, since all angles of \(T\) are at least \(m_n\), the area-diameter Lemma~\ref{lm:area-controls} gives \[ \operatorname{area}(T) \geq \frac12\,\operatorname{diam}(T)^2\sin^2 m_n. \] Combining the two inequalities yields \[ \frac12\,\operatorname{diam}(T)^2\sin^2 m_n \leq q_n^k \operatorname{area}(T_0). \] Therefore \[ \operatorname{diam}(T) \leq \frac{\sqrt{2\,\operatorname{area}(T_0)}}{\sin m_n} q_n^{k/2}. \] Since \(q_n<1\), the right-hand side tends to zero as \(k\to\infty\). \end{proof} \begin{remark} The argument does not require the longest edge of every child to be strictly shorter than the longest edge of its parent. Such a one-step statement can fail, for instance when a child inherits one side of the parent triangle. Instead, the proof uses an explicit uniform area contraction together with the minimum angle condition. \end{remark}

\section{Conclusion}\label{sec:conclusion}

The longest-edge $n$-section and the largest-angle $n$-section algorithms behave differently for large values of $n$.  The known theory of the longest-edge mesh refinements shows that bisection and trisection produce nondegenerate triangular partitions, but the longest-edge $n$-section, for $n\geq 4$,  always produce sequences of degenerating triangles \cite{RosenbergStenger1975,PlazaSuarezPadronFalcon2010,SuarezMorenoAbadPlaza2012,PerdomoPlaza2012,KorotovPlazaSuarez2015}.  In contrast, the largest-angle $n$-section remains regular for every $n\geq 2$.

The reason for such a difference in the performance of the algorithms is essentially based on the following observation.  The angle divided by the largest-angle $n$-section is always at least $\pi/3$. Therefore, each newly created sector angle is at least $\pi/(3n)$.  This gives the uniform lower bound (and also the uniform upper bound) for every descendant angle in terms of the initial smallest angle and $n$.  Together with a compactness-based area contraction argument, this also implies that the recursive largest-angle $n$-sections produce descendants whose all the diameters tend to zero.

Thus, the degeneration observed for the longest-edge $n$-section when $n\geq 4$ is not an inevitable consequence of dividing a triangle into many subtriangles.  It is tied to the longest-edge selection rule.  The largest-angle selection rule avoids this particular instability and provides a regular $n$-section refinement procedure for all $n\geq 2$.

\end{document}